\documentclass{amsart}
\usepackage{graphicx} 
\usepackage{hyperref}
\usepackage{amsmath,amsfonts,amssymb,amsthm,mathrsfs,tikz-cd, enumerate,mathtools,amscd,tikz-cd}

\usepackage{epigraph}
\usepackage[T2A]{fontenc}

\DeclarePairedDelimiterX\set[1]\lbrace\rbrace{#1}

\theoremstyle{plain}
\newtheorem{definition}{Definition}[section]

\newtheorem{theorem}{Theorem}[section]
\newtheorem{prop}{Proposition}[section]

\newtheorem{corollary}{Corollary}[section]

\newcommand{\ra}{\rightarrow}

\newcommand{\RR}{{\mathbb R}}
\newcommand{\NN}{{\mathbb N}}
\newcommand{\ZZ}{{\mathbb Z}}
\newcommand{\CC}{{\mathbb C}}

\newcommand{\FF}{{\mathbb F}}

\newcommand{\SA}{{\mathscr A}}

\newcommand{\M}{{\mathbf M}}
\newcommand{\R}{{\mathbf R}}
\newcommand{\K}{{\mathbf K}}

\newcommand{\tX}{{\tilde X}}
\newcommand{\tZ}{{\tilde Z}}

\title{Chaos and thermalization in Clifford--Floquet  dynamics}
\author{Anton Kapustin}
\address{California Institute of Technology}
\author{Daniil Radamovich}
\address{University of California, Berkeley}

\begin{document}

\begin{abstract}
    We study the ergodic properties of a unitary Floquet dynamics arising from the repeated application of a translationally-invariant Clifford Quantum Cellular Automata to an infinite system of qubits in $d$ dimensions. One expects that if the QCA does not exhibit any periodicity, a generic initial state of qubits will thermalize, that is, approach the infinite-temperature state. We show that this is true for many classes of states, both pure and mixed. In particular, this is true for all initial states that are short-range entangled and close to the equilibrium state. We also point out a subtle distinction between weak and strong thermalization. 
\end{abstract}

\maketitle

\section{Introduction}
\epigraph{Сказать, оно, конечно, всё можно, а ты поди продемонстрируй!}{D. I. Mendeleev}

Ergodic theory characterizes the long time behavior of deterministic classical dynamical systems and identifies those that behave "chaotically" at long time scales. It also provides quantitative measures of deterministic chaos (for example, the Kolmogorov-Sinai entropy). For a readable introduction to ergodic theory see e.g. \cite{CFS}. There are many examples of classical dynamical systems exhibiting deterministic chaos: special types of billiards \cite{billiards}, hyperbolic automorphisms of tori \cite{CFS} and, more generally, Anosov diffeomorphisms and flows \cite{KatokHasselblatt}. Yet another arena for ergodic theory is provided by cellular automata \cite{Pivatoreview}.  

The theory of deterministic quantum chaos is much less developed. Only the lowest rungs of the classical ergodic hierarchy (ergodicity, weak mixing, mixing) have reasonable quantum counterparts \cite{AlickiFannes}. While quantum analogs of the Kolmogorov-Sinai entropy have been defined \cite{AlickiFannes,CNT}, much less is known about them. 

A better understanding of quantum chaos could clarify the foundations of statistical mechanics and thermodynamics. For example, a common (although not unanimous) belief going back to Ludwig Boltzmann and Paul and Tatiana  Ehrenfest is that the tendency of closed systems to thermalize is related to chaotic behavior. Assuming this, the problem is to understand why chaos is a robust property of commonly occurring quantum systems. 

In both physics and mathematics, examples usually come first, and general theory is created as a means of understanding them. To develop a theory of quantum chaos, it would be helpful to have a supply of tractable quantum systems that are universally agreed to be chaotic in an informal sense.

One natural place to look for such examples is the theory of Quantum Cellular Automata (QCAs). A QCA $\alpha$ is a unitary map that acts on a system of qubits (or more generally, qudits) and has good locality properties: if $a$ is an observable localized at a point $j$, $\alpha(a)$ is localized on a ball of radius $R$ and center $j$, where $R$ can be taken the same for all $j$ and $a$. Iterating a QCA gives rise to unitary Floquet dynamics. A general QCA does not have any integrals of motion, hence for a generic initial state, one expects the system to heat up indefinitely, which means that at long times the state of any finite subsystem approaches the infinite-temperature state.\footnote{Here and below we consider infinite systems of qubits. Unitary evolution of finite systems is quasiperiodic, so thermalization of such systems can occur only when they are coupled to a bath.}

To test this, one can begin by looking at an especially tractable type of QCAs: Clifford QCAs. Consider a finite or infinite system of qubits. Its algebra of observables is generated by Pauli matrices $X_j,Z_j$ where the index $j$ labels the qubits. A Clifford QCA is a QCA which maps each Pauli matrix to a monomial in Pauli matrices. One can show that  translationally-invariant (TI) Clifford QCAs are essentially equivalent to a special type of classical Cellular Automata (CA) \cite{CQCA1}, therefore the methods used to study the latter give insight into the behavior of the former. 

The first study of the Floquet dynamics generated by TI Clifford QCAs was undertaken in \cite{CQCA2} and was limited to qubit chains with a single qubit per site. Such TI Clifford QCAs admit a trichotomy into periodic, glider, and fractal types. From the point of view of ergodic theory, fractal QCAs are the most interesting ones. In particular, Ref. \cite{CQCA2} argued that fractal TI Clifford QCAs thermalize arbitrary TI product states. This means that for any local observable $a$ one has
\begin{equation}\label{eq:thermnaive}
    \lim_{n\ra\infty}\langle \alpha^n(a)\rangle_0=\langle a\rangle_\infty 
\end{equation}
where $\langle\ldots\rangle_0$ and $\langle\ldots\rangle_\infty$ are the initial product state and the infinite-temperature state, respectively. This property indicates that fractal TI Clifford QCAs are chaotic quantum dynamical systems.\footnote{The fractal condition is strictly stronger than the mixing condition. As discussed below, there are many TI Clifford QCAs which are mixing and integrable. These examples illustrate that mixing and chaos should not be identified.} Ref. \cite{CQCA2} shows that fractal TI Clifford QCAs also thermalize  arbitrary TI Pauli stabilizer states. We are not aware of any other example of a deterministic many-body system, whether classical or quantum, where thermalization was shown for a comparably large class of initial states. 

In this note we re-examine the Clifford-Floquet dynamics by making a fuller use of the machinery of classical CAs. First of all, we explain how to generalize most of the results of Ref. \cite{CQCA2} to systems with an arbitrary number of qubits per site and to higher spatial dimensions. The key notion is that of a diffusive CA. It was introduced in \cite{PV1,PV2} and further studied in \cite{ETDS}. Diffusive Clifford QCAs  are a generalization of fractal QCAs of Ref. \cite{CQCA2}. 

Second, we greatly enlarge the class of states that can be shown to thermalize under the action of diffusive  Clifford QCAs. In particular, we prove that any short-range entangled state which is sufficiently close to the infinite-temperature state is thermalized. 
This strengthens the case that diffusive Clifford QCAs are examples of chaotic deterministic quantum dynamics. 

Third, we point out a gap in the proof of thermalization of TI product states in \cite{CQCA2}. Specifically, the arguments in Ref. \cite{CQCA2} do not establish that the sequence on the left-hand side of eq. (\ref{eq:thermnaive}) has a limit for an arbitrary fractal QCA. One can prove a weaker statement: a fractal Clifford QCA thermalizes any TI product state for almost all times. This means that the limit exists and is equal to the right-hand side after we exclude a subset of times which is negligible in a suitable sense. We will call this property {\it weak thermalization} as opposed to strong thermalization described by eq. (\ref{eq:thermnaive}). We will see that it is much easier to establish weak thermalization for other classes of states as well.

The paper is organized as follows.  Section 2 is preparatory. We recall the relation between Clifford QCAs and classical cellular automata (CAs) and  discuss Clifford QCAs from the point of view of ergodic theory. We point out that the quantum ergodic hierarchy collapses in this case:  ergodic is equivalent to mixing is equivalent to $r$-mixing for all $r>2$. In Section 3 we explain a gap in the proof of Proposition III.2 in \cite{CQCA2} and show how it can be repaired using the results of \cite{ETDS} at the price of weakening the conclusion. We also discuss the growth of operator support under fractal/diffusive Clifford-Floquet dynamics and show that in some cases the support grows only logarithmically with time. 
In Section 4 we discuss which classes of states are thermalized by diffusive Clifford-Floquet dynamics. This section contains the main new results of the paper. In Section 5 we present numerical evidence that diffusive Clifford QCAs thermalize some states which are not covered by our rigorous results. In Section 6 we discuss our findings.

We are grateful to Bowen Yang and B. Hellouin de Menibus for discussions. This work was supported in part by the U.S.\ Department of Energy, Office of Science, Office of High Energy Physics, under Award Number DE-SC0011632 and by the Simons Investigator Award.

\section{Clifford QCAs and pseudo-unitary CAs}

In this paper we consider translationally-invariant QCAs acting on an infinite chain of qubits, or more generally, on an infinite system of qubits on a hypercubic lattice $\ZZ^d\subset\RR^d$. Since we deal almost exclusively with translationally-invariant QCAs, in what follows we omit the adjective "translationally-invariant". Non-translationally-invariant QCAs appear only in Section 4.

Let us review the connection between {\it Clifford} QCAs and classical CAs following \cite{CQCA1,CQCA2}. Let $N$ be the number of qubits per unit cell and $X_i^\mu,Z_i^\mu$, $i\in\ZZ^d$, $\mu\in \{1,\ldots,N\}$ be the $X$ and $Z$ Pauli matrices acting site-wise. Pauli matrices are generators of the algebra of local observables $\SA$. A basis in this algebra is given by all monomials in Pauli matrices. Such monomials are in one-to-one correspondence with Laurent polynomials in variables 
$u_1,\ldots, u_d$ with coefficients in $\FF_2^{2N}$, where $\FF_2=
\{0,1\}$ is the field with two elements. In what follows we will denote the $d$-tuple $(u_1,\ldots,u_d)$ by $u$. If $k\in\ZZ^d$, then $u^k$ denotes the monomial $u_1^{k_1}\ldots u_d^{k_d}$, while $u^{-1}$ denotes $(u_1^{-1},\ldots,u_d^{-1})$.

Let us denote by $\M$ the set of all Laurent polynomials in $u$ with coefficients in  $\FF_2^{2N}$ and by $P_q$ the Pauli monomial corresponding to a Laurent polynomial $q(u)$. To make this well-defined, one needs to choose an ordering on Pauli matrices, otherwise the map from Laurent polynomials to Pauli monomials is defined only up to a sign. In this paper, such signs are not going to be important, so we do not explicitly specify an ordering. For any $q,q'\in \M$ we have $P_q P_{q'}=\pm P_{q+q'}$.

A Clifford QCA is a QCA $\alpha$ which maps a Pauli monomial to a Pauli monomial. Since we assume translational invariance, the action on the corresponding Laurent polynomials is given by
\begin{align}
    q(u)\mapsto L(u) q(u),
\end{align}
where $L(u)$ is a non-degenerate $2N\times 2N$ matrix whose elements are Laurent polynomials in $u$ with coefficients in $\FF_2$. For example, the QCA which shifts all qubits by one site in the direction $a\in \{1,\ldots,d\}$ corresponds to  $L(u)=u_a\cdot 1_{2N}$. If we choose an ordering convention for Pauli matrices, then $L(u)$ uniquely determines the corresponding Clifford QCA. Thus we have for any $q\in \M$: 
\begin{align}
    \alpha(P_q)=\pm P_{Lq},
\end{align}
where the sign depends on the chosen ordering. 

A matrix $L(u)$ as above also defines an update rule for a classical cellular automaton with $2^{2N}$ states per site labeled by elements of a finite abelian group $\FF_2^{2N}$. Namely, if we encode a point $x\in (\FF_2^{2N})^{\ZZ^d}$ as a formal power series
\begin{equation}
    x(u)=\sum_{k\in\ZZ^d} x_k u^k\in \FF_2^{2N}((u,u^{-1})),\quad x_k\in \FF_2^{2N},
\end{equation}
then the update rule is
\begin{equation}
    x(u)\mapsto x'(u)=\left(L\left(u^{-1}\right)^T\right)^{-1} x(u).
\end{equation}
Since the update rule is linear and invertible, the corresponding CA is linear and invertible. Dually, we can consider an update rule for continuous functions on the compact topological space $\Upsilon=(\FF_2^{2N})^{\ZZ^d}$. Any such function can be expanded in a convergent series in  characters of the abelian group $\Upsilon$, so it is sufficient to describe the update rule for characters. Characters can be identified as Laurent polynomials in $u$ with coefficients in $\FF_2^{2N}$: to every $q=\sum_j q_j u^j\in\M$ one attaches a function
\begin{align}
    \chi_q(x)=\prod_{j\in\ZZ^d} (-1)^{\sum_{\mu=1}^{2N} q_j^\mu x_j^\mu}.
\end{align}
The dual update rule is $\chi_q\mapsto \chi_{L q}$. The QCA/CA correspondence thus maps Pauli monomials to characters of $\Upsilon$.

Since $\alpha$ is an automorphism of the algebra of local observables, the update rule satisfies an additional condition which we call pseudo-unitarity. Let $\R$ be the ring of Laurent polynomials with coefficients in $\FF_2$. One can think of $L$ as an element of the nonabelian group $GL(2N,\R)$ and $\M$ as a free rank-$2N$ module over this ring. Let us define involutions on $\R$ and $\M$ by letting $\bar r(u)=r(u^{-1})$,  $\bar q(u)=q(u^{-1})$ for any $r\in \R$ and any $q\in \M$. Let us define a function $\Omega:\M\times \M\ra \R$ by
\begin{align}
    \Omega(q,q')=\sum_{a,b}\omega^{ab} \bar q_a q'_b,
\end{align}
where $\omega$ is a $2N\times 2N$ matrix with values in $\FF_2$:
\begin{align}
    \omega=
    \begin{pmatrix}
        0 & 1_N \\
        1_N & 0 
    \end{pmatrix}
\end{align}
The function $\Omega$ is $\R$-linear in the second argument: $\Omega(q,rq')=r\Omega(q,q')$ for all $r\in\R$ and all $q,q'\in\M$. It is anti-linear in the first one: $\Omega(r q,q')=\bar r\Omega(q,q')$. It also satisfies $\Omega(\bar q,\bar q')=\overline {\Omega(q',q)}$. We say that $L\in GL(2N,\R)$ is pseudo-unitary if $\Omega(Mq,Mq')=\Omega(q,q')$. It is easy to see that Clifford QCAs are in one-to-one  correspondence with pseudo-unitary elements of $GL(2N,\R)$. This bijection is compatible with the group structure up to signs: if $L,L'$ are pseudo-unitary matrices corresponding to Clifford QCAs $\alpha,\alpha'$, then for any $q\in \M$ we have 
\begin{align}
    (\alpha\circ\alpha')(P_q)=\pm P_{LL'q}. 
\end{align}
In particular, the action of the automorphism $\alpha^n$ on Pauli monomials is determined, up to signs, by the matrix $L^n$. This maps many questions about Clifford QCAs to questions about pseudo-unitary linear CAs. 

As an application of the QCA/CA correspondence, let us prove the following result.
\begin{theorem}
    The following statements are equivalent for a  Clifford QCA.
    \begin{itemize}
        \item The QCA is ergodic.
        \item The QCA is mixing.
        \item The QCA is $r$-mixing for all $r\geq 2$.
        \item The corresponding CA is ergodic.
        \item The corresponding CA is mixing.
        \item The corresponding CA is $r$-mixing for all $r\geq 2$.
        \end{itemize}
\end{theorem}
Here, it is implicit that the state on the quantum spin chain is the infinite-temperature state $\omega_\infty$, while the probability measure for the CA is the normalized Haar measure.\footnote{By a state we mean a positive linear functional on the algebra of observables $\SA$. This covers both pure and mixed states. The value of a state $\omega$  on an observable $a$ will be denoted $\omega(a)$.} The equivalence of the last three conditions is, of course, well known \cite{Rokhlin1949,ShirvaniRogers}. Let us show their equivalence to the first three.

Recall that an automorphism of a $*$-algebra $\SA$ is called 2-mixing (or simply mixing) with respect to a state $\omega$ if it satisfies
\begin{align}\label{eq:mixing}
    \lim_{n\ra +\infty}\omega\left( a_0\,\alpha^n(a_1)\right)=\omega\left(a_0\right)\omega\left(a_1\right)
\end{align}
for all $a_0,a_1\in\SA$. If $\omega\left( a_0\,\alpha^n(a_1)\right)$ converges to $\omega\left(a_0\right)\omega\left(a_1\right)$ in the Ces\`{a}ro sense, then $\alpha$ is ergodic. If there is $r\in\NN$, $r\geq 3$, such that for all $a_0,\ldots,a_{r-1}\in\SA$ we have
\begin{align}\label{eq:rmixing}
    \lim_{n,k_1,\ldots,k_{r-2}\ra +\infty}\omega\left( a_0\alpha^n(a_1)\alpha^{n+k_1}(a_2)\ldots\alpha^{n+k_{r-2}}(a_{r-1})\right)=\prod_{j=0}^{r-1}\omega\left( a_j\right),
\end{align}
then $\alpha$ is called $r$-mixing with respect to $\omega$. Clearly, if $\alpha$ is $r$-mixing, then it is $r'$-mixing for all $r'<r$, as well as ergodic. In the case of interest to us, $\SA$ is the algebra of local observables of the spin chain. It is sufficient to check the $r$-mixing condition for Pauli monomials, since they form a basis. 

Ergodic, mixing, and $r$-mixing CAs are defined similarly, with $\SA$ replaced by the commutative algebra of continuous functions on the configuration space $\Upsilon$ and $\alpha$ replaced with the update rule for such functions. It is sufficient to check the desired properties  for characters of $\Upsilon$, since they form an orthonormal basis in the space of continuous functions on $\Upsilon$. 

For a quantum spin system, the natural state to consider is the infinite-temperature state $\omega_\infty$. It is uniquely determined by the condition $\omega(ab)_\infty=\omega (ba)_\infty$ for all $a,b\in\SA$. Its value on Pauli monomials is
\begin{align}\label{eq:tracialstate}
    \omega(P_q)_\infty=\left\{ \begin{matrix}1, & q=0\\
    0, & q\neq 0\end{matrix}\right.
\end{align}
The $r$-mixing property restricted to Pauli monomials says that the limits on the l.h.s. of eqs. (\ref{eq:mixing}) and (\ref{eq:rmixing}) vanish unless $a_j=1$ for all $j\in\{0,\ldots,r-1\}$.

For a CA, the natural state is the average over the normalized Haar measure on $\Upsilon=(\FF_2^{2N})^{\ZZ^d}$, i.e., the probability measure such that all $x_j$ are independent and each $x_j$ takes each value with equal probability . Its value on characters is
\begin{align}\label{eq:Haar}
    \int_\Upsilon\chi_q(x) d\mu(x)=\left\{ \begin{matrix}1, & q=0\\
    0, & q\neq 0\end{matrix}\right.
\end{align}
The equivalence of the $r$-mixing property (resp.  ergodicity) of a Clifford QCA $\alpha$ and the $r$-mixing property (resp. ergodicity) of the corresponding pseudo-unitary CA follows from the comparison of (\ref{eq:tracialstate}) and (\ref{eq:Haar}).

Armed with Theorem 1, we can now construct many examples of mixing Clifford QCAs. We use the following result from the theory of linear CAs \cite{ShirvaniRogers}.
\begin{theorem}\label{th:mixingCA}
    Let $L(u)\in GL(k,\R)$. The corresponding linear CA is mixing with respect to the Haar measure if and only if $L(u)^n$ has no eigenvectors in $\M=\R^k$ with eigenvalue $1$ for any $n>0$, or equivalently, if $\det(L(u)^n-1)\neq 0$ for all $n>0$. 
\end{theorem}
The condition that $L(u)^n$ has no unit eigenvalues for all $n>0$ has the following interpretation on the QCA side: no Pauli monomial behaves periodically with respect to the Clifford-Floquet dynamics. In Ref. \cite{CQCA2} Clifford QCAs which {\it fail} the mixing/ergodicity test above were called periodic. Thus, for Clifford QCAs "periodic" is synonymous with non-ergodic or non-mixing.

For example, let $d=1$ and $k=2$ and consider a CA  of the form 
\begin{equation}\label{eq:Mexample}
L(u)=\begin{pmatrix}
    0 & 1 \\ 1 & t(u)
\end{pmatrix}
\end{equation}
where $t(u)\in\R$. One can show that it is mixing if $t$ is non-constant, see Appendix. If $t(u)=t(u^{-1})$, then $L(u)$ is pseudo-unitary and thus defines a mixing Clifford QCA (in fact, a Clifford circuit). For example, for $t(u)=u+u^{-1}$ we get the QCA 
\begin{equation}\label{eq:glider}
    \alpha:X_n\ra Z_n,\ Z_n\mapsto X_n Z_{n-1} Z_{n+1},\quad n\in\ZZ,
\end{equation}
while for $t(u)=u+1+u^{-1}$ we get the QCA studied in Ref. \cite{Kentetal}:
\begin{equation}\label{eq:diffusiveexample}
    \alpha:X_n\ra Z_n,\ Z_n\mapsto X_n Z_n Z_{n-1} Z_{n+1},\quad n\in\ZZ.
\end{equation}

A more trivial example in 1d is the "shift QCA" which acts by translations to the left or to the right. That is, $L(u)=u^{\pm 1} 1_{2N}.$ The mixing property of the shift QCA arises from the fact that every Pauli monomial eventually gets shifted "to infinity". This example illustrates that mixing is a necessary but not a sufficient condition for chaos, and that mixing can coexist with integrability. 
A similar mechanism operates in the case of the QCA (\ref{eq:glider}): every Pauli monomial eventually decomposes into a configuration of widely separated "gliders" which have a fixed shape and move with a unit speed to the left or to the right \cite{CQCA2}. 

The mechanism that leads to mixing in the case of the QCA (\ref{eq:diffusiveexample}) is very different. It can be attributed to the unbounded growth of the support of any nontrivial Pauli monomial as a function of time.

\section{Strongly and weakly diffusive QCAs}

A special feature of quantum spin systems (and many-body quantum systems in general) is that they have a natural set of coarse-grained descriptions independent of any dynamics. We simply consider observables supported on a finite subset of the lattice. Larger sets correspond to finer coarse-graining. This suggests that the key to chaotic behavior is an unbounded growth of operator support as a function of time. As time goes by, increasingly finer coarse-graining is required to avoid information loss. To make this precise in the Clifford setting, let us define the Hamming weight of a Laurent polynomial $q\in\M$ to be the number of monomials in $q$. We will denote the Hamming weight of $q$ by $|q|$. It is a measure of the size of the support of the observable $P_q$.
\begin{definition}
    Let $L\in GL(2N,\R)$. $L$ is strongly diffusive if for any  non-zero $q\in \M$ we have $\lim_{n\ra\infty}|L^nq|=\infty$.  A Clifford QCA or a linear CA is called strongly diffusive if the corresponding matrix $L$ is strongly diffusive. 
\end{definition}
Strongly diffusive CAs were introduced in \cite{PV1}, where they are called diffusive CAs.

It is easy to see that every strongly diffusive Clifford QCA is mixing with respect to the infinite-temperature state. Indeed, if $a,b$ are Pauli monomials and $b$ is nontrivial, then $\lim_{n\ra\infty}| a\alpha^n(b)|=\infty$, hence there exists $n_0$ such that for all $n\geq n_0$ $a\alpha^n(b)$ is a nontrivial Pauli monomial and thus $\omega(a\alpha^n(b))_\infty=0$. On the other hand, in the previous section we saw examples of mixing Clifford QCAs which preserve the Hamming weight of some or all Pauli monomials. Thus strong diffusivity is strictly stronger than mixing. Similarly, every strongly diffusive CA is mixing with respect to the Haar state.

In general, it seems difficult to check whether a given linear CA or a Clifford QCA is strongly diffusive. There is a weaker property that is easier to check \cite{PV1}. A subset $J$ of the set of natural numbers is said to have density $1$ if $\lim_{n\ra\infty} \frac{1}{n}|J\cap [1,n]|=1$. The complement of a density-1 subset is called a zero-density subset.  Intuitively, a density-1 subset comprises almost all natural numbers, while a zero-density subset contains a vanishingly small fraction of naturals. For example, the set of all prime numbers is zero-density, and so is the set of all complete squares.
\begin{definition}\label{th:diffusiveindensity}
    $L\in GL(r,\R)$ is called weakly diffusive if for any non-zero $q\in\R^r$ there is a density-1 subset $J_q\subset \NN$ such that 
\begin{equation}
    \lim_{n\in J_q,n\ra\infty}|L^n q|=\infty.
\end{equation} 
A Clifford QCA or a linear CA are called weakly diffusive if the corresponding $L$ is weakly diffusive. 
\end{definition}
In Ref. \cite{PV1} weakly diffusive CAs are called diffusive in density.

In plain English, weakly diffusive dynamics has the property that the Hamming weight of a nontrivial  observable diverges if one excludes a vanishingly small fraction of times (a zero-density subset of $\NN$). We will show in the next section that weakly diffusive Clifford QCAs thermalize a large class of states for almost all times. 

Ref. \cite{ETDS} showed that weakly diffusive CAs admit an algebraic characterization.\footnote{The proof given in \cite{ETDS} is formulated for CAs in one spatial dimension, but it can be easily generalized to arbitrary $d$. We are grateful to B. Hellouin de Menibus for pointing this out.}
\begin{theorem}\label{th:diffusiveindensityalg}
A linear CA is weakly diffusive if and only if the equation $L^n q=u^k q$ for $q\in\M$ does not have non-zero solutions for all $n>0$ and all $k\in\ZZ^d$. Equivalently, the matrix $L$ satisfies $\det(L^n-u^k)\neq 0$ for all $n>0$ and all $k\in\ZZ$. 
\end{theorem}
Non-zero solutions of $L^nq=u^kq$ are characters  of $\Upsilon$ which are $n$-periodic in time up to a spatial translation by $k$ units. In Ref. \cite{ETDS} such characters are called {\it solitons}. From the QCA point of view, a soliton corresponds to a Pauli monomial $q$ that is invariant under a combination of some iterate of the QCA $\alpha$ and a spatial translation. Pairs $(n,k)\in \ZZ\times\ZZ^d$ label $\ZZ$-subgroups of the symmetry $\ZZ\times\ZZ^d$ of the problem (time translations times spatial translations). The above theorem shows that a Clifford QCA $\alpha$ is weakly diffusive  if and only if no local quantity is preserved by a $\ZZ$-symmetry. 

Ref. \cite{CQCA2} studied $d=1$, $N=1$ Clifford QCAs. It calls a non-zero solution of the equation $Lq=u^k q$ a {\it glider} for $M$. Thus saying that $M$ does not have solitons is equivalent to saying that $L^n$ does not have gliders for any $n>0$. Ref. \cite{CQCA2} called such Clifford QCAs fractal. Thus, soliton-free Clifford QCAs are a generalization of fractal Clifford QCAs to general $N$ and $d$. Comparing Theorem \ref{th:diffusiveindensity} and Theorem \ref{th:mixingCA}, we also see that every weakly diffusive Clifford QCA is mixing. 

Using Theorem \ref{th:diffusiveindensityalg}, it is easy to construct examples of weakly diffusive Clifford QCAs. For example, the 1d Clifford QCA whose matrix $L$ has the form (\ref{eq:Mexample}) for some nonconstant palindromic Laurent polynomial $t\in\R$ is weakly diffusive if and only if $t(u)\neq u^m+u^{-m}$ for some $m\in\NN$. For a proof, see Appendix. In the exceptional case $t(u)=u^m+u^{-m}$ one can explicitly construct a soliton: the vector $q(u)=(1,u^m)$ solves the equation $L q=u^m q$.

Prop. III.2 in \cite{CQCA2} states that if $M$ is fractal, then it strongly thermalizes any translationally-invariant product state. The proof hinges on the claim that $\lim_{n\ra\infty}|L^nq|=\infty$ for any non-zero $q\in\M$, i.e. on the strong diffusivity of $L(u)$. By  Theorem \ref{th:diffusiveindensityalg}, the fractal condition is equivalent to weak diffusivity. Thus, the proof of Prop. III.2 can go through as stated only if every weakly diffusive $L$ is strongly diffusive, at least in the case of a single qubit per site and $d=1$.

The argument for strong diffusivity in \cite{CQCA2} involves two steps. First, the sequence $|L^n q|$ is shown to be unbounded for any $q\neq 0$, i.e., $\limsup_n |L^n q|=\infty$. Second, it is argued (by referring to the proof of Prop. III.1) that an infinite sequence of times with a bounded $|L^n q|$ forces some power of $L$ to have a glider, leading to a contradiction and the conclusion that $\liminf_n |L^n q|=\infty$. The first step is valid (the unboundedness of the sequence $|L^nq|$ also follows from weak diffusivity). However, the second step appears to be faulty. The quickest way to see this is to consider the following linear CA with a single bit per site:
\begin{equation}
    x'_n=x_{n-1}+x_{n+1},\quad n\in\ZZ.
\end{equation}
The corresponding $1\times 1$ matrix $L$ is $L=u+u^{-1}\in\R$. It is easy to see that such  $L$ does not admit solitons, hence it is weakly diffusive. On the other hand, for any integer $m$ we have
\begin{equation}
L^{2^m}=u^{2^m}+u^{-2^m}.
\end{equation}
Hence $|L^{2^m}q|=2|q|$ for any $q$ and all sufficiently large $m$ which implies that $L$ is not strongly diffusive.\footnote{While this CA is not pseudo-unitary and not even invertible, and does not correspond to a Clifford QCA, the argument for the existence of a glider in \cite{CQCA2} does not use pseudo-unitarity or invertibility.}

Lacking a proof of strong diffusivity for a general soliton-free Clifford QCA, one may ask if there are any examples of Clifford QCAs which are provably strongly diffusive. Ref. \cite{ETDS} constructs two examples of invertible 1d CAs which are strongly diffusive. These CAs involve two bits per site and are not pseudo-unitary. But one can turn them into pseudo-unitary CAs by the following "doubling trick". For any $A\in GL(k,\R)$, we can define a pseudo-unitary element $L\in GL(2k,\R)$ by letting
\begin{align}
 L(u)=   \begin{pmatrix}
    A(u) & 0 \\ 0 & \left(A(u^{-1})^T\right)^{-1}.
\end{pmatrix}
\end{align}
It is clear that doubling a strongly diffusive CA gives a strongly diffusive pseudo-unitary CA. Ref. \cite{ETDS} shows that the following two 1d CAs are strongly diffusive:
\begin{equation}\label{eq:FG}
    F(u)=\begin{pmatrix} 0 & 1 \\ 1 & u\end{pmatrix},\quad G(u)=\begin{pmatrix} 0 & 1 \\ 1 & 1+u\end{pmatrix}.
\end{equation}
Therefore the following two 1d Clifford QCAs with two qubits per site  are strongly diffusive:
\begin{align}
\alpha_F: X_n\mapsto \tX_n,\ \tX_n\mapsto X_n \tX_{n+1},\ Z_n\mapsto \tZ_n Z_{n+1},\ \tZ_n\mapsto Z_n.
\end{align}
and
\begin{align}
\alpha_G: X_n\mapsto \tX_n,\ \tX_n\mapsto X_n \tX_n\tX_{n+1},\ Z_n\mapsto Z_n \tZ_n Z_{n+1},\ \tZ_n\mapsto Z_n.
\end{align}
Here, the pairs $X_n,Z_n$  and $\tX_n,\tZ_n$ describe the two qubits on the $n^{\rm th}$ site.

We could not find examples of Clifford QCAs (equivalently, pseudo-unitary CAs) which are weakly diffusive but not strongly diffusive. It might be that reversibility together with weak diffusivity imply strong diffusivity. For illustration purposes, in Fig. 1-3 we plotted the Hamming weight of $L^nq$ as a function of $n$ for the weakly diffusive $L$ of the form (\ref{eq:Mexample}) with $t(u)=u+1+u^{-1}$. We let $q$ be constant, so that $P_q$ is a single-site Pauli monomial, but we checked that the results for two-site and three-site Pauli monomials are similar. We see that although the Hamming weight has very large fluctuations, it appears to diverge as $n\ra\infty$. Note that the plots for $P_q=X_0$ and $P_q=Z_0$ appear identical. When plotted with a finer resolution, one finds that they are related by a spatial translation $j\mapsto j+1$. This is an accidental feature special to this example.

\begin{figure}
    \centering
    \includegraphics[width=1\linewidth]{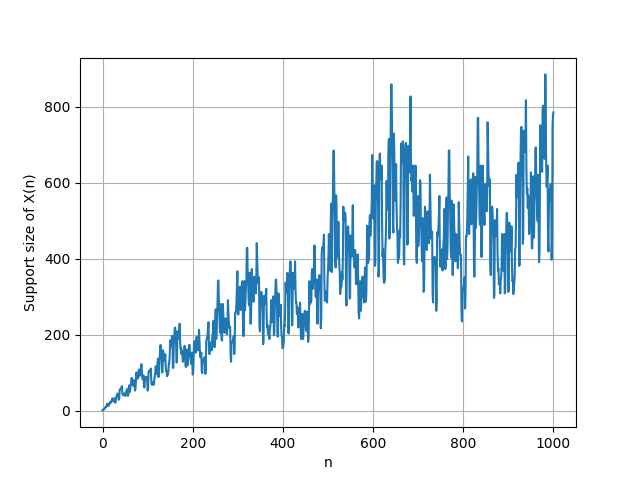}
    \caption{Support size of $X(n)=\alpha^n(X_0)$ as a function of $n$ for the QCA (\ref{eq:Mexample}) with $t(u)=u+1+u^{-1}$.}
    \label{fig:Fig1}
\end{figure}

\begin{figure}
    \centering
    \includegraphics[width=1\linewidth]{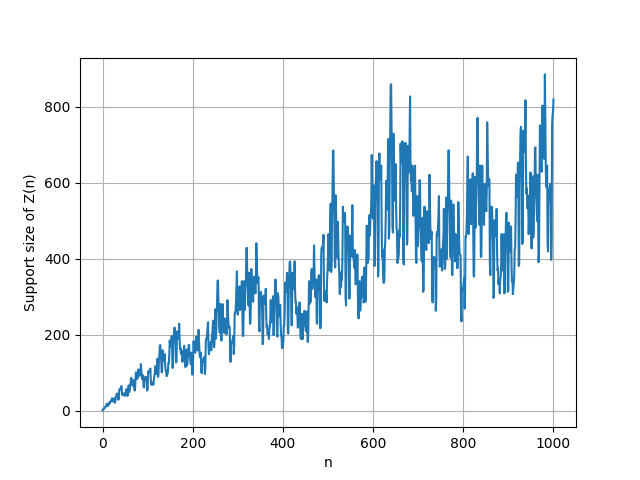}
    \caption{Support size of $Z(n)=\alpha^n(Z_0)$ as a function of $n$ for the QCA (\ref{eq:Mexample}) with $t(u)=u+1+u^{-1}$.}
    \label{fig:Fig2}
\end{figure}

\begin{figure}
    \centering
    \includegraphics[width=1\linewidth]{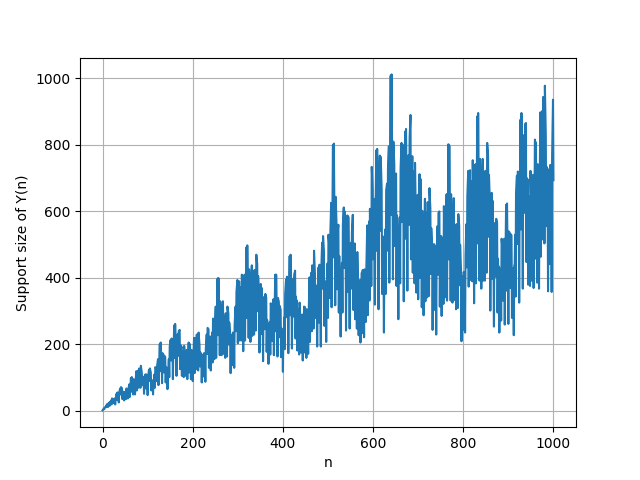}
    \caption{Support size of $Y(n)=\alpha^n(Y_0)$ as a function of $n$ for the QCA (\ref{eq:Mexample}) with $t(u)=u+1+u^{-1}$.}
    \label{fig:Fig3}
\end{figure}

For a strongly diffusive $L$, it is natural to inquire about the rate of growth of $|L^nq|$ as a function of $n$. Since $L$ has a finite range, $|L^nq|$ is upper-bounded by a linear function of $n$. On the other hand, we see that in general there are large downward fluctuations in $|L^nq|$, so a lower bound on $|L^nq|$ is likely to grow much slower. For concreteness, let $L(u)=F(u)$ where $F(u)$ is given by (\ref{eq:FG}). One can easily show that 
\begin{align}
    F^{2^r}=\begin{pmatrix}
        \sum_{j=1}^r u^{2^r-2^j} & u^{2^r-1}\\
        u^{2^r-1}  & u^{2^r}+\sum_{j=1}^r u^{2^r-2^j}.
    \end{pmatrix}
\end{align}
Thus $|F^{2^r}q|\leq (2r+1)|q|$. This implies that any lower bound on $|F^nq|$ cannot grow faster than $O(\log_2 n)$.

\section{Some classes of thermalized states}

Ref. \cite{CQCA2} studied 1d Clifford QCAs with a single qubit per unit cell and argued that two classes of TI states are strongly thermalized by any soliton-free QCA of this sort: arbitrary TI Pauli stabilizer states and arbitrary TI product states.\footnote{Ref. \cite{CQCA2} also shows that glider Clifford QCAs thermalize arbitrary TI Pauli stabilizer states.} The proof of the latter relies on strong diffusivity. As we saw in the previous section, there is a gap in the proof of Prop. III.2, hence strong thermalization of arbitrary product states cannot be inferred.

Instead, we have the following weaker result.
\begin{theorem}\label{th:fractaltherm}
Let $\alpha$ be a soliton-free 1d Clifford QCA and $\omega$ be a TI product state of a qubit chain (one qubit per site). Then $\alpha$ weakly thermalizes $\omega$. That is, for any local observable $a$ there is a density-1 subset $J_a\subset\NN$ such that
\begin{align}
         \lim_{n\in J_a} \omega(\alpha^n(a))=\omega_\infty( a). 
    \end{align}
\end{theorem}
\begin{proof}
    If one of the numbers $|\omega(X_0)|,|\omega(Y_0)|,|\omega(Y_0)|$ is equal to $1$, then $\omega$ is a Pauli stabilizer state and is thermalized in the strong sense (Section III.C of \cite{CQCA2}). Thus, it is sufficient to consider the case when $|\omega(P)|<1$ for any on-site Pauli monomial $P$. Let $\lambda<1$ be the maximal value of $|\omega(P)|$ where $P\in \{X_0,Y_0,Z_0\}$. By Theorem \ref{th:diffusiveindensityalg}, $L$ is weakly diffusive, hence for any non-zero $q\in\M$ there is a density-1 $J_q\subset\NN$ such that 
    \begin{align}
         \lim_{n\in J_q} \lambda^{|L^n q|}=0. 
    \end{align}
    This implies $\lim_{n\in J_q} |\omega(\alpha^n(P_q))|=0$ for any non-zero $q$. Finally, any local observable $a$ can be written as  $\sum_{q\in K}  c_q P_q$ where the set $K$ is finite. Therefore, 
    \begin{align}
         \lim_{n\in J_a} \omega(\alpha^n(a))=c_0=\omega_\infty(a). 
    \end{align}
     where $J_a=\bigcap_{q\in K} J_{q_k}$ is density-1.
    \end{proof}

From a practical standpoint, weak thermalization is almost as good as strong thermalization. For example, it is easy to see that it implies convergence of time-averages of expectation values to their equilibrium values. That is, if a state $\omega$ is weakly thermalized by an automorphism $\alpha$, then for any local observable $a$ we have
\begin{align}
    \lim_{M\ra\infty} \frac{1}{M}\sum_{m=0}^{M-1}\omega(\alpha^{mT}(a))=\omega_\infty(a),
\end{align}
where $T\in\NN$ is arbitrary.

The above result can be generalized in various directions. Consider a soliton-free Clifford QCA in dimension $d$ and with an arbitrary number $N$ of qubits per site. First of all, one can prove that such a QCA weakly thermalizes a generic product state. To make this  precise, we introduce the following definition.
\begin{definition}
    For any state $\omega$ on a system of qubits, let $\lambda(\omega)=\sup_{P} |\omega(P)|$,
    where the supremum is taken over the set of all nontrivial Pauli monomials. $\omega$ is called P-generic if $\lambda(\omega)<1$.  
\end{definition}
Note that for a product state $\omega$ of a lattice system of qubits, P-genericity is equivalent to the requirement that there exists $\lambda<1$ such that $|\omega(P_q)|\leq\lambda$ for any on-site Pauli monomial $P_q$. Both pure and mixed states can be P-generic. For example, for a single qubit any pure state which is not an eigenstate of $X,Y$, or $Z$ is P-generic. A sufficient condition for P-genericity is being not too far from the infinite-temperature state. More precisely, we have a trivial 
\begin{prop}\label{th:trivial}
    If $\|\omega-\omega_\infty\|\leq \lambda<1$, then $\omega$ is P-generic and $\lambda(\omega)\leq\lambda$. 
\end{prop}

The same argument as in Theorem \ref{th:fractaltherm} proves 
\begin{theorem}\label{th:Pgeneric}
    Let $\alpha$ be a soliton-free Clifford QCA and let $\omega$ be a P-generic product state for a system of qubits in $d$ dimensions ($N$ qubits per unit cell). Then $\alpha$ weakly thermalizes $\omega$. If, in addition, $\alpha$ is strongly diffusive, then $\alpha$ strongly thermalizes $\omega$.
\end{theorem}

As in \cite{CQCA2}, one can generalize these results by replacing P-generic product states by states of the form $\omega=\omega_0\circ\beta$ where $\omega_0$ is a P-generic product state and $\beta$ is an arbitrary Clifford QCA. 
\begin{theorem}
    Let $\alpha$ be a soliton-free Clifford QCA, $\beta$ be an arbitrary Clifford QCA, and $\omega_0$ be a P-generic product state. Then $\alpha$ weakly thermalizes $\omega_0\circ\beta$. If $\alpha$ is strongly diffusive, then it strongly thermalizes $\omega_0\circ\beta$.
\end{theorem}
\begin{proof}
    First, note that $\omega_0\circ\beta$ is strongly (resp. weakly) thermalized by a TI Clifford QCA $\alpha$ if and only if $\omega_0$ is strongly (resp. weakly) thermalized by $\beta\circ\alpha\circ\beta^{-1}$. Second, suppose $M$ and $M'$ be two pseudo-unitary elements of $GL(2N,\R)$. Then $M'M {M'}^{-1}$ is strongly (resp. weakly) diffusive if $M$ is strongly (resp. weakly) diffusive. Combining this with Theorem \ref{th:diffusiveindensityalg} and Theorem \ref{th:Pgeneric}, we get the desired result.
\end{proof}

Even more generally, consider a state of the form $\omega=\omega_0\circ\beta$ where $\omega_0$ is a P-generic product state and $\beta$ is an arbitrary (not necessarily translationally-invariant or Clifford) QCA. Such a state $\omega$ may be called short-range entangled (SRE), since correlations vanish beyond a certain range $r$ (the range of the QCA $\beta$).\footnote{Our terminology is somewhat non-standard here. Usually, an SRE state is defined as a state of the form $\omega_0\circ\beta$ where $\omega_0$ is a {\it pure} product state and $\beta$ is a finite-depth unitary circuit.} We will now prove that a soliton-free Clifford QCA thermalizes any SRE state which is sufficiently close to the equilibrium state. 
\begin{theorem}\label{th:SREtherm}
    Consider a $d$-dimensional system of qubits with $N$ qubits per unit cell. Let $\alpha$ be a soliton-free Clifford QCA and $\beta$ be any (not necessarily translationally-invariant or  Clifford) QCA. There exists $C_\beta>1$ that  depends only on $d$, $N$ and the range of $\beta$ such that $\alpha$ weakly thermalizes any state of the form $\omega=\omega_0\circ\beta$ where $\omega_0$ is a product state with  
    \begin{align}
        \lambda(\omega_0) C_\beta<1. 
    \end{align}
     If, in addition, $\alpha$ is strongly diffusive, then it strongly thermalizes such an $\omega$. 
\end{theorem}
\begin{proof}
    First, we need to clarify the notion of the range of a QCA. Let us fix some distance function on $\ZZ^d$. Apart from the usual Euclidean distance, one can use the "Manhattan distance"  $d_1(x,y)=\sum_{i=1}^d |x_i-y_i|$ or the sup-distance $d_\infty(x,y)={\rm max}_i |x_i-y_i|$. For any $\Gamma\subset \ZZ^d$ and any $r\geq 0$ we denote by $\Gamma^r$ the $r$-thickening of $\Gamma$, i.e. the set of all lattice points which are within distance $r$ from $\Gamma$. For example, $\{j\}^r$ is the ball of radius $r$ and center $j$. Note that $|\Gamma^r|\leq |\Gamma|\cdot V_r$, where $V_r$ is the number of lattice points in $\{j\}^r$.
    
    We say that $\beta$ has range less or equal than $r$ if for any $j\in\ZZ^d$ and any observable $a$ localized on $j$ $\beta(a)$ is localized on $\{j\}^r$. The range of $\beta$ is the smallest $r$ with this property. If $\beta$ has range $r$ and $a$ is any observable localized on $\Gamma$, then $\beta(a)$ is localized on $\Gamma^r$. Note also that the range of $\beta^{-1}$ is equal to the range of $\beta$. 

    For any non-zero $q\in\M$ we denote by $\Gamma_q\subset\ZZ^d$ the support of the local observable $P_q$, so that $|\Gamma_q|=|q|$. We can expand $\beta(P_q)$ as a finite sum
    \begin{align}\label{eq:betasum}
    \beta(P_q)=\sum_{q'\in\M\backslash \{0\}} c_q^{q'} P_{q'}.
    \end{align}
    The number of terms in this sum does not exceed the number of nontrivial Pauli monomials localized on $\Gamma^r_q$, i.e. $2^{2N|\Gamma_q^r|}-1$ which is upper-bounded by $2^{2NV_r|q|}$. 
    We claim that 
    \begin{align}\label{eq:l1est}
        \sum_{q'} \left|c^{q'}_q\right|\leq 2^{N V_r|q|}
    \end{align}
    Indeed, the Pauli monomials $P_q$ form an orthonormal basis in the space of local observables with respect to the scalar product
    $\langle a,b\rangle=\omega_\infty(a^*b)$. Since $\beta$ preserves this scalar product, we have
    \begin{align}
        \sum_{q'}\left|c^{q'}_q\right|^2=1. 
    \end{align}\
    Applying Cauchy-Schwarz, we get (\ref{eq:l1est}). 

    We can also lower-bound the Hamming weight of Laurent polynomials $q'$ appearing in the sum (\ref{eq:betasum}) in terms of the Hamming weight of $q$. Indeed, since $c^{q'}_q=(P_{q'},\beta(P_q))=(\beta^{-1}(P_{q'}),P_q)$, for $c^{q'}_q$ to be non-zero we must have $\Gamma_q\subseteq \Gamma^r_{q'}$, hence $|q|\leq V_r|q'|$. 
    
    Since $\omega_0$ is a product state, for any local observable $a=\sum_q c_q P_q$ we have
    \begin{align}\label{eq:est2}
        \left|\omega_0(a)\right|\leq \sum_q |c_q| \lambda(\omega_0)^{|q|}.
    \end{align}
    Combining (\ref{eq:l1est}) and (\ref{eq:est2}) and using $|q'|\geq |q|/V_r$ we get 
    \begin{align}
        |\omega_0\circ\beta(\alpha^n(P_q))|\leq \mu^{|L^n q|},
    \end{align}
    where $L\in GL(2N,\R)$ is the matrix defining $\alpha$ and
    \begin{align}
        \mu=\lambda(\omega_0)^{\frac{1}{V_r}} 2^{N V_r}.
    \end{align}
    Thus if we let 
    \begin{align}
        C_\beta=2^{N V_r^2},
    \end{align}
    we get the desired result.
\end{proof}

\begin{corollary}
    Let $\omega$ be an SRE state, i.e. $\omega=\omega_0\circ\beta$ for some (not necessarily TI or Clifford) QCA $\beta$ and some product state $\omega_0$. There exists a constant $C_\beta>1$ which depends only on the range of $\beta$, $d$, and $N$ such that 
    \begin{align}
        \|\omega-\omega_\infty\|< 1/C_\beta
    \end{align}
    implies that $\omega$ is weakly (resp. strongly) thermalized by any weakly (resp. strongly) diffusive  Clifford QCA. 
\end{corollary}
\begin{proof}
    Follows from Proposition \ref{th:trivial} and the fact that $\|\omega-\omega_\infty\|=\|\omega_0-\omega_\infty\|$.
\end{proof}
Note that the larger is the range of $\beta$, the larger is $C_\beta$, and thus the closer the state $\omega$ should be to $\omega_\infty$ to guarantee thermalization. 

One can enlarge the set of thermalized states even further. Let $\omega$ be any state thermalized by a Clifford QCA $\alpha$. Performing a local von Neumann measurement on $\omega$ results in a state $\omega'$ which has the form
\begin{align}
    \omega'(a)=\frac{\omega(b a b)}{\omega(b)},
\end{align}
where $b=b^*$ is a projector in the algebra of local observables. We can expand $b$ as a finite sum, $b=\sum_{i\in I} c_i P_{q_i}$, where $c_i,i\in I$ are real numbers. Then for any $q\in\M$ we have
\begin{equation}
    \omega'(\alpha^n(P_q))=\frac{1}{\omega(b)}\sum_{i,j\in I}\pm c_i c_j\omega(P_{q_i+q_j+L^n q}).
\end{equation}
If $L$ is strongly diffusive and $q\neq 0$, all terms on the right-hand side tend to zero as $n\ra\infty$, hence the state $\omega'$ is strongly thermalized. For a weakly diffusive Clifford QCA, the same argument works if we exclude some zero-density subset of times.

Finally, let us briefly discuss the thermalization of TI Pauli stabilizer states. These states are associated to Lagrangian (i.e. maximal isotropic) $\R$-submodules of $\M$ and defined by the condition $\omega(P_q)=\pm 1$ whenever $q$ lies in the chosen Lagrangian submodule. For $N>1$, soliton-free Clifford QCA may fail to thermalize some TI Pauli stabilizer states. For example, a Clifford QCA obtained from a CA by the doubling trick always has an  invariant Lagrangian submodule. If $q$ belongs to this submodule and $\omega$ is a Pauli stabilizer state corresponding to this submodule, then $\omega(\alpha^n(P_q))=\pm 1$ for all $n$, hence $\alpha$ does not thermalize $\omega$, whether weakly or strongly. To ensure thermalization of arbitrary Pauli stabilizer states, one needs to require that $\alpha$ does not have invariant Lagrangian submodules. 

\section{Numerical results}

By Theorem \ref{th:SREtherm}, diffusive Clifford QCAs thermalize short-range entangled states which are sufficiently close to the equilibrium state. It is natural to wonder if the condition of being sufficiently close to $\omega_\infty$ can be dropped. In this section we explore  this question numerically for a particular class of examples. 

\begin{figure}
    \centering
    \includegraphics[width=1\linewidth]{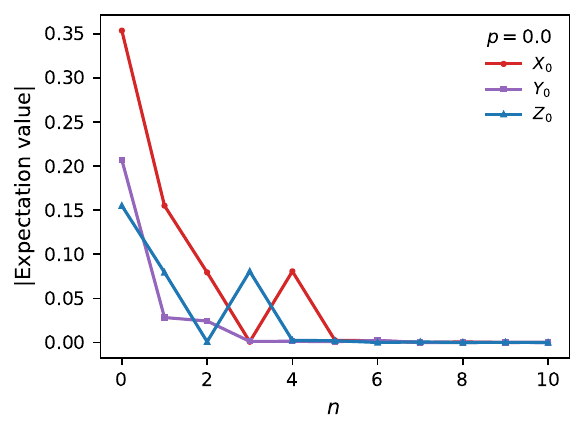}
    \caption{Absolute values of expectation values $\omega_0\circ\beta(\alpha^n(P_q))$ of single-site Pauli monomials as functions of $n$. The product state $\omega_0$ corresponds to $p=0$, $\theta=30^{\circ}$, $\phi=45^{\circ}.$ The non-Clifford QCA $\beta$ corresponds to (\ref{eq:nonClbeta}) with $r=1$.}
    \label{fig:Fig4}
\end{figure}

\begin{figure}
    \centering
    \includegraphics[width=1\linewidth]{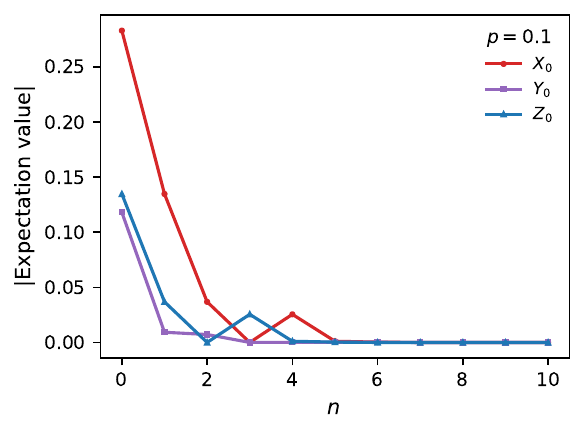}
    \caption{Absolute values of expectation values $\omega_0\circ\beta(\alpha^n(P_q))$ of single-site Pauli monomials as functions of $n$. The product state $\omega_0$ corresponds to $p=0.1$, $\theta=30^{\circ}$, $\phi=45^{\circ}.$ The non-Clifford QCA $\beta$ corresponds to (\ref{eq:nonClbeta}) with $r=1$.}
    \label{fig:Fig5}
\end{figure}

\begin{figure}
    \centering
    \includegraphics[width=1\linewidth]{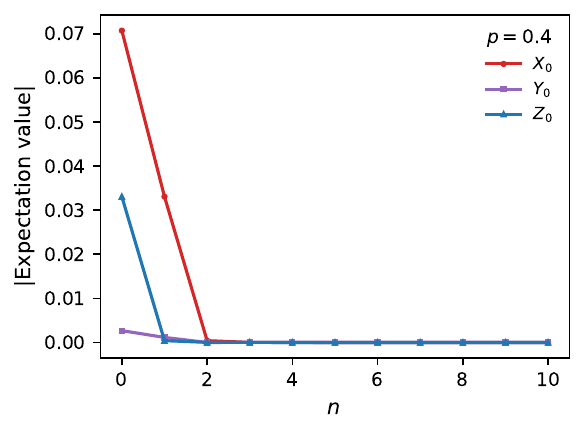}
    \caption{Absolute values of expectation values $\omega_0\circ\beta(\alpha^n(P_q))$ of single-site Pauli monomials as functions of $n$. The product state $\omega_0$ corresponds to $p=0.4$, $\theta=30^{\circ}$, $\phi=45^{\circ}.$ The non-Clifford QCA $\beta$ corresponds to (\ref{eq:nonClbeta}) with $r=1$. }
    \label{fig:Fig6}
\end{figure}

We let $\alpha$ be the Clifford QCA with the matrix $L$ of the form (\ref{eq:Mexample}) with $t(u)=u+u^{-1}+1$. This QCA is weakly diffusive by Theorem \ref{th:diffusiveindensityalg} and Appendix A. We take $\omega_0$ to be a translationally-invariant product state corresponding to an on-site density matrix $\rho$. $\rho$ is specified by a choice of  eigenvalues $p$ and $1-p$, $p\in [0,0.5)$, and a point on the Bloch sphere parameterized by spherical coordinates $(\theta,\phi)$. Finally, we take a translationally-invariant but non-Clifford QCA $\beta$ of the form 
\begin{equation}
    \beta= \prod_{j\in\ZZ} {\rm Ad}_{U_{j,j+1}},
\end{equation}
where $U_{j,j+1}$, $j\in\ZZ$, is a unitary matrix acting in the Hilbert space of two qubits on sites $j$ and $j+1$. We choose $U_{j,j+1}$ so that it commutes with its translates. As a concrete example of such matrices, we consider
\begin{equation}\label{eq:nonClbeta}
    U^{(1)}_{j,j+1}=\exp(i r X_j\otimes X_{j+1}),
\end{equation}
where $r\in [0,2\pi)$ is a parameter.  Note that $\beta$ has range $1$.

We computed the absolute value of $\omega_i(P_{L^n q})=\omega_0\circ\beta (P_{L^n q})$ as a function of $n$, $n\leq 10$, for several choices of the Laurent polynomial $q$ and parameters of $\omega_0$. The computation algorithm is described in Appendix B. We considered three values of $p$: $p=0$, $p=0.1$, and $p=0.4$. The first one corresponds to a pure initial state, the remaining two to mixed initial states. In Fig. 4-6 we display the results for $q$ corresponding to on-site Pauli monomials $X_0,Y_0,Z_0$ for some representative values of $r$ and angles $(\theta,\phi)$. We see that in all cases the convergence to equilibrium is quite rapid. Changing the parameters gives qualitatively similar results. We also tested some other values for the initial Laurent polynomial $q$ and obtained very similar results. Finally, we extended the simulation to $n\leq 1000$ and found that the expectation values remain extremely close to the equilibrium values. Thus, the numerical evidence suggests that thermalization persists for generic values of all parameters.

\section{Discussion}

In this paper, we examined quantum chaotic behavior arising in Floquet quantum spin systems without disorder. This is to be contrasted with quantum spin systems driven by random unitary circuits (see \cite{randomquantumcircuits} for a review) or spatially random circuits \cite{randomTP1,randomTP2}. We generalized the results of \cite{CQCA2} to Clifford QCAs in general dimensions and an arbitrary number of qubits per site. For simplicity, we stated our results for systems of qubits, but identical results hold for systems of qudits. We showed that soliton-free Clifford QCAs (i.e. those Clifford QCAs which do not preserve any local quantity up to translation) thermalize a large class of states, including all SRE states sufficiently close to the equilibrium state and states obtained from these by local von Neumann measurements. Moreover, our numerical simulations indicate that soliton-free Clifford QCAs thermalize a larger class of SRE states which are quite far from equilibrium and are not covered by our rigorous results. All of this evidence supports the identification of soliton-free Clifford QCAs as  examples of chaotic deterministic quantum dynamics. 

We still lack a precise criterion for what counts as deterministic chaos in the quantum case. Such a definition should be possible for infinite many-body quantum systems with sufficiently local dynamics, because such systems have a distinguished coarse-graining procedure.


Throughout this paper, we emphasized the distinction between strong and weak thermalization, i.e. between thermalization for all times and for almost all times. From the experimental standpoint, weak thermalization is difficult to distinguish from strong thermalization, but is easier to prove mathematically.

\appendix

\section{Mixing and weak diffusivity of particular CAs}

\begin{prop}
    Let \begin{equation}\label{eq:Lspecial}
L(u)=\begin{pmatrix}
    0 & 1 \\ 1 & t(u)
\end{pmatrix}
\end{equation}
where $t(u)\in\R$. $L(u)$ is mixing if and only if $t(u)$ is not a constant.
\end{prop}

\begin{proof}
We need to prove that $\det(L(u)^n-1)\neq 0$ for all $n>0$ if $t(u)$ is not constant. For $n\geq 0$ we have 
\begin{equation}\label{eq:Mn}
    L(u)^n=\begin{pmatrix} b_{n-1} & b_n \\ b_n & b_{n+1}\end{pmatrix}
\end{equation}
where $b_n$, $n=0,1,\ldots$ is a sequence of Laurent polynomials uniquely defined by $b_0=0$, $b_1=1$, and 
\begin{equation}\label{eq:recursion}
    b_{n+1}= t b_n+b_{n-1}.
\end{equation} 
Also, since $\det L^n=1$ for all $n$, we have 
\begin{equation}\label{eq:Cassini}
    b_{n-1} b_{n+1}+b_n^2=1.
\end{equation}
Hence $\det (L^n-1)=b_{n-1}+b_{n+1}$. For any Laurent polynomial $p$, let $\deg_{\min} p$ (resp. $\deg_{\max}(p))$ be the smallest negative (resp. greatest positive) power of $u$ that occurs in $p(u)$. Then $\deg_{\min}(b_n)=(n-1)\deg_{\min}(t)$ and $\deg_{\max}(b_n)=(n-1)\deg_{\max}(t)$ and thus $b_{n-1}+b_{n+1}\neq 0$ for all $n>0$ if either $\deg_{\min}(t)\neq 0$ or $\deg_{\max}(t)\neq 0$.
\end{proof}

\begin{prop}
    The matrix $L$ given by (\ref{eq:Lspecial}) for some nonconstant palindromic Laurent polynomial $t\in\R$ is weakly diffusive if and only if $t(u)\neq u^m+u^{-m}$ for some $m\in\NN$.
\end{prop}
\begin{proof}
Let $t(u)\in \R$ be a nonconstant palindromic Laurent polynomial, so that $L(u)$ is a pseudo-unitary element of $GL(2,\R)$. By Theorem \ref{th:diffusiveindensityalg}, it is sufficient to show that the equation
\begin{equation}
    L^nq=u^kq,\quad q\in \M
\end{equation}
for $q$ has non-zero solutions for some $n>0$ and some $k\in\ZZ$ if and only if $t(u)=u^m+u^{-m}$ for some $m\neq 0$. In fact, we will prove a slightly more general statement: the equation
\begin{equation}
    L^n q=r q,\quad q\in \M
\end{equation}
for $q$ has non-zero solutions for some $n>0$ and some Laurent polynomial $r$ if and only if $t(u)=u^m+u^{-m}$. 

Eq. (\ref{eq:Mn}) and eq. (\ref{eq:Cassini}) imply that 
\begin{equation}\label{eq:req}
    \det (L^n-r)=r^2+r t b_n+1.
\end{equation}
Let $t b_n=s_n$. We claim that any solution $r\in\R$ of $r^2+s_n r+1=0$ must be of the form $r=u^m$ for some $m\in\ZZ$. Indeed, on the one hand, $r$ divides $r^2+1$. On the other hand, $r^2+1=1 {\rm mod} r$, so ${\rm gcd}(r,1+r^2)=1$ (here we used that $\R$ is a Principal Ideal Domain). Hence $r$ must divide $1$, i.e. $r$ is a unit in the ring $\R$. Thus $r=u^k$ for some $k\in\ZZ$. 

The equation $r^2+s_n r+1=0$ is equivalent to  $s_n=u^k+u^{-k}$ where $s_n=t b_n$ and $b_n$ is determined by the recursion (\ref{eq:recursion}). We claim this implies $n$ divides $k$ and $t=u^{k/n}+u^{-k/n}$. Indeed, consider the equation 
\begin{equation}
    \alpha^2+ \alpha t+1=0,
\end{equation}
so that $t=\alpha+\alpha^{-1}$. Roots of this equation belong to the algebraic closure $\overline\R$ of $\R$. If $\alpha$ is one root, than $\alpha^{-1}$ is the other. Note now that $s_n=t b_n$ satisfies the same recursive equation as $b_n$, but with different initial conditions: $s_0=0$, $s_1=t$. Induction on $n$ shows that $s_n=\alpha^n+\alpha^{-n}$. Thus $\alpha^n$ satisfies the same quadratic equation as $r=u^k$, and therefore either $\alpha^n=u^k$ or $\alpha^n=u^{-k}$.

We have shown that $\alpha^n$, which a priori lies in the algebraic closure of $\R$, in fact lies in the field of fractions $\K$. This forces $\alpha$ itself to lie in $\K$. Indeed, if it were not in $\K$, then the Galois involution would exchange $\alpha$ and $\alpha^{-1}$ and $\alpha^n$ and $\alpha^{-n}$. But since $\alpha^n\in\K$, it must be Galois-invariant, hence $\alpha^{2n}=1$. Therefore $\alpha$ lies in the algebraic closure of $\FF_2$. But this contradicts the fact that $\alpha+\alpha^{-1}=t$ is a non-constant element of $\R$. Hence $\alpha\in \K$. This argument also shows that for a non-constant Laurent polynomial $t$ the integer $k$ must be non-zero.

We showed that $\alpha$ is a ratio of two polynomials with coefficients in $\FF_2$. On the other hand, $\alpha^n=u^{\pm k}$ where $k$ is non-zero. Looking at the behavior near $u=0$ and $u$ we see that $k$ must be divisible by $n$ and $\alpha$ must have the same order poles or zeros at $u=0$ and $u=\infty$ as $u^{\pm k/n}$. Hence $\alpha=u^{\pm k/n}$ and $t=u^{k/n}+u^{-k/n}$, as claimed.

In the opposite direction, if $t(u)=u^m+u^{-m},$ it is easy to check that $q(u)=(1,u^m)$ solves $Lq=u^mq$.
\end{proof}

\section{Evaluating averages of Pauli monomials in short-range entangled states}

We would like to evaluate state averages of the form
\begin{equation}
    \omega\left(\beta(P_{q_1}P_{q_2}\ldots P_{q_K})\right),
\end{equation}
where $P_{q_i}$ is an on-site Pauli monomial on site $i$, $\beta$ is a QCA of range $r$ (possibly non-Clifford and not translationally-invariant), and $\omega$ is a product state (possibly not translationally-invariant). The naive evaluation algorithm has a run-time that is exponential in $K$. Here we describe an algorithm with a run-time that is linear in $K$. 

First, we reduce the problem to the case $r=1$. We combine sets of $r$ consecutive sites into supersites, so that in terms of supersites, $\beta$ has range $1$. If the original problem has $N$ qubits per site, then each supersite has $rN$ qubits. 

From now on, we assume $r=1$ and $N$ qubits per site. Let $A$ be the algebra of observables on a single site. It has dimension $4^N$. We pick a basis $e_a$, $a=1,\ldots, 4^N,$ for $A$ so that $e_1$ is the unit of $A$ and the remaining $e_a$ are nontrivial Pauli monomials. We let $f^a_{bc}$ be the structure constants of $A$ in this basis, i.e.
\begin{equation}
    e_b e_c=\sum_a f^a_{bc} e_a.
\end{equation}

To simplify notation, we assume that $\omega$ and $\beta$ are translation-invariant. Then $\omega$ is completely described by $4^N$ real numbers  $\omega_a=\omega(e_a)$. $\beta$ is fully specified by its values on single-site Pauli monomials. Let $e^i_a$ be the Pauli monomial corresponding to $e_a$ which is localized on site $i$. Then 
\begin{equation}
    \beta(e^i_a)=\sum_{b,c,d}\beta^{bcd}_a e^{i-1}_b\otimes e^i_c\otimes e^{i+1}_d,
\end{equation}
and $\beta$ is completely determined by the $4^{4N}$ complex numbers $\beta^{bcd}_a$.

We define
\begin{equation}
    M^{kl}_{a,ij}=\sum_{b,c,d}\omega_d f^d_{ib} f^k_{jc} \beta^{bcl}_a
\end{equation}
We regard $M$ as the data of $4^N$ linear maps $M_a:A\otimes A\rightarrow A\otimes A$ defined by
\begin{equation}
    M_a(e_i\otimes e_j)=\sum_{k,l}M^{kl}_{a,ij} e_k\otimes e_l
\end{equation}
Finally, the expectation values of Pauli monomials on two sites in the state $\omega$ are described by a linear map $\gamma:A\otimes A\ra \CC$ whose components are 
\begin{equation}
    \gamma(e_i\otimes e_j)=\omega_i\omega_j. 
\end{equation}

Using this notation, the state average for $q_i=e_{a_i}$ can be computed as 
\begin{equation}
    \gamma M_{a_K} \ldots M_{a_2} M_{a_1}(e_1\otimes e_1).
\end{equation}
This expression can be viewed as a component of a Matrix Product State with the bond space $A\otimes A$ and the physical space $A$. Equivalently, it is an output a $\CC$-weighted automaton whose space of states is $A\otimes A$, the start vector is $e_1\otimes e_1$, the final dual vector is $\gamma$, and the transition matrix on an input $e_a$ is $M_a$. 

The same algorithm works when $\omega$ is an arbitrary product state and $\beta$ is an arbitrary QCA. The run-time is still linear in $K$. However, it grows exponentially with the range of $\beta$.



\bibliographystyle{unsrt} 
\bibliography{Bibliography.bib} 

 \end{document}